\documentclass[11pt]{article}
\pdfoutput=1

\usepackage{amsthm}
\usepackage{amsfonts}
\usepackage{pifont}

\usepackage{xspace,enumerate}
\usepackage{amsmath,amssymb}
\usepackage[toc,page]{appendix}
\usepackage{thmtools}
\usepackage{thm-restate}
\usepackage{graphicx}
\usepackage{csquotes}
\usepackage{makecell}
\usepackage{tabularx}
\usepackage{relsize}
\usepackage{scalerel}
\usepackage[dvipsnames]{xcolor} 
\usepackage{tikz}
\usepackage{tcolorbox}
\tcbuselibrary{skins,breakable}
\tcbset{enhanced jigsaw}
\usepackage[linesnumbered,ruled,vlined]{algorithm2e}
\usepackage{nicefrac}
\usepackage{mathrsfs}

\usepackage[top=1in, bottom=1in, left=1in, right=1in]{geometry}
\usepackage[pdfstartview=FitH,pdfpagemode=UseNone,colorlinks,linkcolor=NavyBlue,filecolor=blue,citecolor=OliveGreen,urlcolor=NavyBlue,pagebackref]{hyperref}

\usepackage[capitalise,nameinlink]{cleveref}
\usepackage[T1]{fontenc}
\usepackage{mathtools,dsfont}

\newtheorem{lemma}{Lemma}[section]
\newtheorem{theorem}[lemma]{Theorem}

\newtheorem{proposition}[lemma]{Proposition}
\newtheorem{observation}[lemma]{Observation}

\usepackage[T1]{fontenc}
\usepackage[tt=false,type1=true]{libertine}
\usepackage[varqu]{zi4}
\usepackage[libertine]{newtxmath}
\newcommand{\LCS}{\mathsf{LCS}}

\title{Ulam Median is NP-hard for Four Permutations{}\thanks{This paper has been superseded by~\cite{azgor2026ulamrankaggregationhard}, which presents a merged version of this work and~\cite{azgor2026hardnessapproximationrankaggregation}.}}
\author{
Mursalin Habib\thanks{{\tt Department of Computer Science, Rutgers University}. {\tt mursalin.habib@rutgers.edu}. Supported by the National Science Foundation under Grants CCF-2313372 and CCF-2443697.}}
\date{}

\begin{document}

\maketitle

\begin{abstract}
We show that computing a median under the Ulam distance is NP-hard even when the input consists of exactly four permutations. Previously, NP-hardness was known only for an unbounded number of input permutations (Fischer \textit{et al.}, ESA '25). Our result is \emph{tight}, since an Ulam median of three permutations can be computed in polynomial time (Chakraborty--Das--Krauthgamer, SODA '21).
\end{abstract}

\section{Introduction}

Suppose several judges rank the same set of candidates, but no two
rankings quite agree. How should these rankings be combined into a
single consensus? This question of \textit{rank aggregation} appears in many places including voting and social
choice theory~\cite{BrandtCELP16}, machine learning~\cite{LiuLQML07},
genomic data analysis~\cite{Li19}, and recommender systems~\cite{OliveiraDLMP20}. A natural answer is to choose the ranking that is ``collectively closest''
to the inputs. More generally, given points \(p_1,\ldots,p_k\) in a
metric space \((\mathcal{M},d)\), their \emph{median} is a point
\(p^\star\in\mathcal{M}\) minimizing
\[
    \sum_{i=1}^k d(p^\star,p_i).
\]
Thus, once a notion of distance between rankings has been fixed, the
median gives a canonical notion of consensus: it is the ranking that
minimizes the total disagreement with the input rankings
\cite{Kemeny59,YoungL78,Young88,DworkKNS01}.

There are several natural ways to measure the distance between two
rankings. Two of the most prominent ones are Kendall's tau distance and the Ulam
distance. Kendall's tau distance counts the number of pairs of items
whose relative order differs between two rankings. Its computational
complexity is by now well understood
\cite{DworkKNS01,FaginKS03}. Dwork \textit{et al}.~\cite{DworkKNS01} showed more
than two decades ago that computing a Kendall-tau median is NP-hard
even for four input permutations. Two very recent works, by Peters~\cite{Peters26} and
Madarasi~\cite{Madarasi26}, prove NP-completeness for exactly
three input permutations, thereby resolving the remaining fixed-input
case.

In this paper, we focus on the Ulam distance. The Ulam distance between
two permutations is the minimum number of \emph{relocations} required
to transform one into the other, where a relocation removes one symbol
and reinserts it at an arbitrary position. Equivalently, for permutations
\(\pi\) and \(\sigma\), their Ulam distance, denoted by \(d_U(\pi, \sigma)\), is:
\[
    d_U(\pi,\sigma)=n-\LCS(\pi,\sigma),
\]
where \(\LCS(\pi,\sigma)\) denotes the length of a longest common
subsequence. In contrast to Kendall's tau distance, which charges for
every inverted pair, the Ulam distance charges only once for an item
that is moved arbitrarily far in the ranking. It therefore provides a
natural measure when a small number of items may be substantially
misplaced
\cite{CormodeMS01,ChakrabortyDK21,ChakrabortyGJ21,ChakrabortyD0S22}.
The Ulam metric is also closely related to the edit metric on arbitrary
strings, a fundamental similarity measure with applications in
computational biology, DNA storage, pattern recognition, and
classification
\cite{Gusfield1997,Pevzner00,GoldmanBCDLSB13,RashtchianMRAJY17,
Kohonen85,Martinez-HinarejosJC00}.

\begin{sloppypar}
The complexity of the Ulam median has had a somewhat unusual history:
approximation algorithms were developed before the problem was even
shown to be NP-hard. Chakraborty, Das, and
Krauthgamer~\cite{ChakrabortyDK21,Chakraborty0K23} initiated the study
of polynomial-time approximation algorithms for the problem and
obtained an approximation factor strictly below \(2\). They also gave
a polynomial-time algorithm for computing an \textit{exact} Ulam median of three
input permutations~\cite{ChakrabortyDK21}. 
NP-hardness of the general problem was established only later by
Fischer, Goldenberg, Habib, and Karthik~\cite{fischer2025hardness}. 
    
\end{sloppypar}

The reduction in~\cite{fischer2025hardness}, however, uses an unbounded
number of input permutations and therefore does not establish hardness
for any fixed number of inputs.\footnote{More precisely, for a
\textsc{Max-Cut} instance \(G=(V,E)\), the construction
in~\cite{fischer2025hardness} produces \(4|V||E|\) input permutations
of length \(3|V|+2\). Thus, as a function of the permutation length
\(N\), the number of inputs is \(O(N^3)\), but it is not bounded by any
constant.} This leaves a substantial
gap between the known polynomial-time algorithm for three permutations
and the general NP-hardness result. In principle, there could still be
polynomial-time algorithms for four permutations, five permutations,
or indeed for every fixed number of input permutations. This is in
sharp contrast with Kendall's tau distance, for which hardness with a
constant number of inputs has been known since 2001. We are thus led to
the central question of this work.

\begin{center}
    {\emph{What is the smallest number of input permutations\\ for
    which computing an Ulam median is NP-hard?}}
\end{center}
In this paper, we give a tight answer to this question.

\begin{theorem}
\label{thm:main}
Computing an Ulam median is NP-hard even when the input
consists of exactly four permutations.
\end{theorem}

Together with the
polynomial-time algorithm for three permutations
\cite{ChakrabortyDK21}, \Cref{thm:main} completely resolves the
complexity threshold with respect to the permutation count: four is the
smallest number of input permutations for which the problem is
NP-hard.

Our proof is a simple, direct reduction from \textsc{3-SAT}. At a high
level, the reduction constructs two permutations that enforce
consistency among the truth values assigned to different occurrences
of the same variable, and two permutations that ensure that every
clause contains a true literal. A single long anchor block of auxiliary symbols separates
these two roles and forces an optimal median to encode a satisfying
assignment. The construction and its analysis rely only on elementary
properties of longest common subsequences.

Finally, our result exposes a sharp distinction between median problems
under the Ulam and edit metrics. Pairwise Ulam distance is closely
related to the edit distance restricted to
permutations. Nevertheless, an Ulam median must itself be a
permutation: every symbol must occur exactly once. This restriction is
algorithmically consequential. For the edit metric on arbitrary strings, an exact median of any fixed
number of input strings can be computed in polynomial time by a
standard multidimensional dynamic program; see, for example,~\cite{NicolasR03}. In contrast, \Cref{thm:main} shows that requiring the median
to be a permutation makes the problem NP-hard already for four inputs.

\paragraph{Other related work.}
Building on~\cite{ChakrabortyDK21,Chakraborty0K23}, subsequent work
developed faster and more general approximation algorithms with factors
strictly below \(2\)~\cite{JaiswalKY25,CarmelDN26}. The best currently
known approximation factor for Ulam median is \(1.968\), due to Carmel,
Das, and Nguyen~\cite{CarmelDN26}. Complementing this line of work, Bai
\textit{et al}.~\cite{BaiFFGMW26} initiated a systematic study of the
parameterized complexity of Ulam median and center, together with their
\(k\)-median and \(k\)-center clustering generalizations.

\section{Preliminaries}
\label{sec:preliminaries}

For a positive integer \(n\), let \([n]:=\{1,2,\ldots,n\}\). All
alphabets considered in this paper are finite. Given an alphabet
\(\Sigma\), a \emph{permutation of\/ \(\Sigma\)} is a string in which
every symbol of \(\Sigma\) appears exactly once. 

We write strings by concatenating their symbols. More generally, the
concatenation of two strings \(\pi\) and \(\sigma\) is denoted simply
by \(\pi\sigma\). The length of a string \(\pi\) is denoted by
\(|\pi|\), and its reversal is denoted by \(\pi^{\mathcal R}\).

Given a string \(\pi\) over an alphabet \(\Sigma\) and a subset
\(S\subseteq\Sigma\), we write \(\pi|_S\) for the restriction of
\(\pi\) to \(S\), obtained by deleting from \(\pi\) every symbol not
contained in \(S\). In particular, if \(\pi\) is a permutation of
\(\Sigma\), then \(\pi|_S\) is a permutation of \(S\).

A string \(\tau\) is a \emph{subsequence} of a string \(\pi\) if it
can be obtained from \(\pi\) by deleting zero or more symbols without
changing the relative order of the remaining symbols. A
\emph{common subsequence} of two strings \(\pi\) and \(\sigma\) is a
string that is a subsequence of both. We denote by
\(\LCS(\pi,\sigma)\) the length of a longest common subsequence of
\(\pi\) and \(\sigma\).

A \emph{symbol relocation} deletes one symbol from a permutation and
reinserts it at an arbitrary position. The \emph{Ulam distance}
between two permutations \(\pi\) and \(\sigma\) of the same alphabet
\(\Sigma\), denoted by \(d_U(\pi,\sigma)\), is the minimum number of
symbol relocations required to transform \(\pi\) into \(\sigma\). We
will frequently use the standard identity
\begin{equation}
    d_U(\pi,\sigma)=|\Sigma|-\LCS(\pi,\sigma).
    \label{eq:ulam-lcs}
\end{equation}

Consequently, minimizing a sum of Ulam distances is equivalent to
maximizing the corresponding sum of longest-common-subsequence
lengths.

Given permutations \(\pi_1,\ldots,\pi_k\) of a common alphabet
\(\Sigma\), an \emph{Ulam median} is a permutation \(\tau^\star\) of
\(\Sigma\) minimizing
\(
    \sum_{i=1}^k d_U(\tau,\pi_i)
\)
over all permutations \(\tau\) of \(\Sigma\). Equivalently, by
\Cref{eq:ulam-lcs}, an Ulam median maximizes
\(
    \sum_{i=1}^k \LCS(\tau,\pi_i).
\)

For disjoint sets \(A\) and \(B\), we write
\(A\mathbin{\dot\cup}B\) for their disjoint union. Throughout the
reduction, all input permutations are defined over the same alphabet,
even when they are described as concatenations of blocks over
pairwise disjoint subalphabets.

\section{The Reduction}

In this section, we prove~\Cref{thm:main}. We give a reduction from \textsc{3-SAT}. Let
\[
\varphi := C_1\wedge C_2\wedge\cdots\wedge C_m
\]
be a \(3\)-CNF formula on the variable set
\(V=\{x_1,x_2,\ldots,x_n\}\), with \(m\geq 2\) clauses. We may further
assume that every clause contains exactly three literals and that, for
each variable \(x\in V\), the number of positive occurrences of \(x\)
equals the number of negative occurrences of \(x\). This restricted
version of \textsc{3-SAT} remains NP-hard~\cite{BermanKarpinskiScott2003}.\footnote{In fact, Berman, Karpinski, and
Scott~\cite{BermanKarpinskiScott2003} establish NP-hardness for the
more restricted problem \((3,B2)\)-\textsc{SAT}, in which every literal
occurs exactly twice.}

Our construction will contain one symbol of the alphabet for each \emph{literal occurrence}
in \(\varphi\). In particular, two appearances of the same literal will
correspond to distinct symbols. Let \(\Sigma\) denote the set of all
literal occurrences and set \(s:=|\Sigma|\). For each variable
\(x\in V\), let \(\Sigma_x^+\) and \(\Sigma_x^-\) denote its positive
and negative occurrences, respectively. For each clause \(C_j\), let
\(\Sigma(C_j)\subseteq\Sigma\) denote its three literal occurrences. We will reason about the satisfiability of \(\varphi\) using the language of \textit{truth partitions} as defined below.

\paragraph{Truth partitions.}
A partition \(\Sigma=T\mathbin{\dot\cup}F\) is \emph{consistent} if,
for every variable \(x\in V\), either
\[
\Sigma_x^+\subseteq T
\quad\text{and}\quad
\Sigma_x^-\subseteq F,
\]
or
\[
\Sigma_x^-\subseteq T
\quad\text{and}\quad
\Sigma_x^+\subseteq F.
\]
It is \emph{satisfying} if \(T\cap\Sigma(C_j)\neq\varnothing\) for every
\(j\in[m]\).

\begin{observation}
\label{obs:truth-partition}
The formula \(\varphi\) is satisfiable if and only if\/ \(\Sigma\) admits
a consistent and satisfying truth partition.
\end{observation}

\begin{proof}
Given a satisfying assignment, let \(T\) contain the occurrences of
literals that evaluate to true and let \(F\) contain those that evaluate
to false. The resulting partition is consistent, and it is satisfying
because every clause contains a true literal.

Conversely, suppose that \(\Sigma=T\mathbin{\dot\cup}F\) is consistent
and satisfying. For each variable \(x\), set \(x=\mathsf{true}\) if
\(\Sigma_x^+\subseteq T\), and set \(x=\mathsf{false}\) otherwise.
Consistency ensures that every occurrence in \(T\) evaluates to true.
Since \(T\) intersects every clause, this assignment satisfies
\(\varphi\).
\end{proof}

We now construct two pairs of permutations. The first pair checks whether the proposed set of true literals contains a witness from every clause. The second pair checks that, for each variable, the proposed set of false literals includes either all of its
positive occurrences or all of its negative occurrences. An anchor block of auxiliary symbols will then force the
median to separate the true literals from the false ones.

For the remainder of the article, fix an arbitrary total ordering of the symbols in \(\Sigma\).

\paragraph{Clause-witness gadget.}
For each clause \(C_j\), let \(A_j\in\Sigma^3\) be the string obtained
by listing the three symbols in \(\Sigma(C_j)\) according to the fixed
ordering of \(\Sigma\). Let \(A_j^{\mathcal R}\) denote its reversal,
and define
\[
\beta_1:=A_1A_2\cdots A_m,
\qquad
\beta_2:=A_1^{\mathcal R}A_2^{\mathcal R}\cdots A_m^{\mathcal R}.
\]

The basic property of this gadget is
\(\LCS(\beta_1,\beta_2)=m\): a common subsequence may contain at most
one symbol from each clause and may choose one symbol from every clause.
We will use the following stronger form where we restrict to an arbitrary subset of all literal occurrences.

\begin{lemma}
\label{lem:clause-gadget}
For every \(S\subseteq\Sigma\),
\[
\LCS(\beta_1|_S,\beta_2|_S)
=
\bigl|\{j\in[m]:S\cap\Sigma(C_j)\neq\varnothing\}\bigr|.
\]
In particular, \(\LCS(\beta_1|_S,\beta_2|_S)=m\) if and only if \(S\)
contains at least one literal occurrence from every clause.
\end{lemma}

\begin{proof}
For each \(j\in[m]\), the relative order of the symbols in
\(\Sigma(C_j)\) is reversed between \(\beta_1\) and \(\beta_2\).
Therefore, a common subsequence of \(\beta_1|_S\) and \(\beta_2|_S\)
can contain at most one symbol from \(S\cap\Sigma(C_j)\). Hence
\[
\LCS(\beta_1|_S,\beta_2|_S)
\leq
\bigl|\{j\in[m]:S\cap\Sigma(C_j)\neq\varnothing\}\bigr|.
\]

Conversely, for every \(j\in[m]\) such that
\(S\cap\Sigma(C_j)\neq\varnothing\), choose one symbol from this
intersection. Since the clauses appear in the same order in
\(\beta_1\) and \(\beta_2\), the chosen symbols, listed in increasing
order of \(j\), form a common subsequence. This proves the reverse
inequality and hence the claimed equality.
\end{proof}

\paragraph{Variable-consistency gadget.}
For each variable \(x\in V\), let \(B_x^+\) and \(B_x^-\) be the
strings obtained by listing the symbols in \(\Sigma_x^+\) and
\(\Sigma_x^-\), respectively, according to the fixed ordering of
\(\Sigma\). Define
\[
\alpha_1
:=
B_{x_1}^+B_{x_1}^-
B_{x_2}^+B_{x_2}^-
\cdots
B_{x_n}^+B_{x_n}^-,
\]
\[
\alpha_2
:=
B_{x_1}^-B_{x_1}^+
B_{x_2}^-B_{x_2}^+
\cdots
B_{x_n}^-B_{x_n}^+.
\]

Because the positive and negative blocks of each variable appear in
opposite orders, a common subsequence of \(\alpha_1\) and \(\alpha_2\) can use symbols from at most one
polarity block for each variable. Since the two blocks have equal
length, this immediately suggests that
\(\LCS(\alpha_1,\alpha_2)=s/2\). We prove the following stronger fact.

\begin{lemma}
\label{lem:consistency-gadget}
For every \(S\subseteq\Sigma\),
\[
\LCS(\alpha_1|_S,\alpha_2|_S)
=
\sum_{x\in V}
\max\bigl\{
|S\cap\Sigma_x^+|,
|S\cap\Sigma_x^-|
\bigr\}.
\]
Consequently,
\(\LCS(\alpha_1|_S,\alpha_2|_S)\leq s/2\), with equality if and only
if, for every variable \(x\in V\), the set \(S\) contains all positive
occurrences of \(x\) or all negative occurrences of \(x\).
\end{lemma}

\begin{proof}
Fix \(x\in V\). Every symbol of \(B_x^+\) precedes every symbol of
\(B_x^-\) in \(\alpha_1\), while every symbol of \(B_x^-\) precedes
every symbol of \(B_x^+\) in \(\alpha_2\). Therefore, a common
subsequence contains symbols from at most one of these two blocks. Its
contribution from the \(x\)-gadget is consequently at most
\(\max\{|S\cap\Sigma_x^+|,|S\cap\Sigma_x^-|\}\).

This bound can be attained simultaneously for every variable: for each
\(x\), choose the larger of \(S\cap\Sigma_x^+\) and
\(S\cap\Sigma_x^-\), list its symbols in their common internal order,
and concatenate these choices in variable order. This proves the
formula.

Let \(r_x:=|\Sigma_x^+|=|\Sigma_x^-|\). We have \(\sum_x r_x=s/2\). Thus,  \(\LCS(\alpha_1|_S,\alpha_2|_S)\leq s/2\) and equality holds precisely when, for every
\(x\), one of \(S\cap\Sigma_x^+\) and \(S\cap\Sigma_x^-\) has size
\(r_x\), which is equivalent to \(S\) containing either all positive
occurrences of \(x\) or all negative occurrences of \(x\).
\end{proof}

\paragraph{Adding an anchor block.}
Let \(\Sigma_b:=\{z_1,z_2,\ldots,z_M\}\) be an alphabet disjoint from
\(\Sigma\), where \(M:=2s+1\), and define
\(Z:=z_1z_2\cdots z_M\). We refer to \(Z\) as the \textit{anchor block} and its symbols as \textit{anchor symbols}. Finally, we construct the four input permutations
\[
L_1:=\beta_1Z,
\qquad
L_2:=\beta_2Z,
\qquad
R_1:=Z\alpha_1,
\qquad
R_2:=Z\alpha_2.
\]
The following lemma shows that the anchor block acts as a separator:
every candidate median may be transformed, without worsening any of
its four LCS scores, into one that places every literal occurrence
entirely to one side of \(Z\).

\begin{lemma}[Anchor-block lemma]
\label{lem:anchor-block}
Let \(\gamma_1,\gamma_2,\delta_1,\delta_2\) be arbitrary permutations
of\/ \(\Sigma\). For every permutation \(\pi\) of\/
\(\Sigma\mathbin{\dot\cup}\Sigma_b\), there exist a partition
\(\Sigma=S_L\mathbin{\dot\cup}S_R\) and permutations \(\tau_L,\tau_R\)
of \(S_L,S_R\), respectively, such that
\(\pi^\star:=\tau_LZ\tau_R\) satisfies
\[
\LCS(\pi^\star,\gamma_iZ)\geq\LCS(\pi,\gamma_iZ)
\quad\text{and}\quad
\LCS(\pi^\star,Z\delta_i)\geq\LCS(\pi,Z\delta_i)
\]
for each \(i\in\{1,2\}\).
\end{lemma}

We defer the proof of \Cref{lem:anchor-block} to
\Cref{sec:anchor-block-proof}. Note that the lemma is completely independent of the
\textsc{3-SAT} instance \(\varphi\): it applies to any four
permutations of \(\Sigma\).  It allows us to restrict attention to permutations of the form
\(\tau_LZ\tau_R\) without loss of generality when analyzing candidate medians of
\(L_1,L_2,R_1,R_2\). Each such ``canonical'' permutation induces a partition
\(\Sigma=S_L\mathbin{\dot\cup}S_R\), where \(S_L\) and \(S_R\) are the
symbols placed before and after \(Z\), respectively. Once this
partition is fixed, the only remaining choice is the ordering of the
symbols within \(S_L\) and \(S_R\). The following lemma determines the
maximum possible sum of the four LCS lengths for a fixed partition.

\begin{lemma}[Fixed-partition optimum]
\label{lem:canonical-score}
Fix a partition
\(\Sigma=S_L\mathbin{\dot\cup}S_R\). Then
\[
\max_{\tau_L,\tau_R}
\sum_{i=1}^2
\Bigl(
\LCS(\tau_LZ\tau_R,L_i)
+
\LCS(\tau_LZ\tau_R,R_i)
\Bigr)
=
4M+s
+\LCS(\beta_1|_{S_L},\beta_2|_{S_L})
+\LCS(\alpha_1|_{S_R},\alpha_2|_{S_R}),
\]
where the maximum ranges over all permutations \(\tau_L\) of \(S_L\)
and \(\tau_R\) of \(S_R\).
\end{lemma}

\begin{proof}
Let \(\pi=\tau_LZ\tau_R\). Since \(M>s\), a longest common subsequence
of \(\pi\) and \(L_i=\beta_iZ\) cannot contain a symbol from
\(\tau_R\): any such subsequence would contain no anchor symbol and
would therefore have length at most \(s<M\), whereas \(Z\) itself is a
common subsequence of length \(M\). Hence
\[
\LCS(\pi,L_i)=M+\LCS(\tau_L,\beta_i|_{S_L}).
\]
Symmetrically,
\(\LCS(\pi,R_i)=M+\LCS(\tau_R,\alpha_i|_{S_R})\).

It remains to optimize the orders \(\tau_L\) and \(\tau_R\). We use the
following elementary identity: if \(\gamma_1,\gamma_2\) are
permutations of the same alphabet \(S\), then
\[
\max_{\tau}
\bigl(
\LCS(\tau,\gamma_1)+\LCS(\tau,\gamma_2)
\bigr)
=
|S|+\LCS(\gamma_1,\gamma_2),
\tag{\(\star\)}
\]
where the maximum ranges over all permutations \(\tau\) of \(S\). It follows from the triangle inequality for Ulam distance, and
equality is attained by taking \(\tau=\gamma_1\).

Applying \((\star)\) to the two sides independently gives
\[
\max_{\tau_L}
\sum_{i=1}^2\LCS(\tau_L,\beta_i|_{S_L})
=
|S_L|+\LCS(\beta_1|_{S_L},\beta_2|_{S_L}),
\]
and
\[
\max_{\tau_R}
\sum_{i=1}^2\LCS(\tau_R,\alpha_i|_{S_R})
=
|S_R|+\LCS(\alpha_1|_{S_R},\alpha_2|_{S_R}).
\]
Adding these equalities and using \(|S_L|+|S_R|=s\) proves the lemma.
\end{proof}



We can now establish the correctness of the reduction.

\begin{proposition}
\label{prop:reduction-correctness}
The formula \(\varphi\) is satisfiable if and only if there exists a
permutation \(\pi^\star\) of the alphabet
\(\Sigma\mathbin{\dot\cup}\Sigma_b\) such that
\[
\sum_{\rho\in\{L_1,L_2,R_1,R_2\}}
d_U(\pi^\star,\rho)
\leq
\frac{5s}{2}-m.
\]
\end{proposition}

\begin{proof}
Suppose first that \(\varphi\) is satisfiable. By
\Cref{obs:truth-partition}, there exists a consistent and satisfying
truth partition \(\Sigma=T\mathbin{\dot\cup}F\). Since \(T\) contains
a literal occurrence from every clause, \Cref{lem:clause-gadget} gives
\(\LCS(\beta_1|_T,\beta_2|_T)=m\). Moreover, consistency implies that,
for every variable \(x\), the set \(F\) contains either all positive
occurrences of \(x\) or all negative occurrences of \(x\). Hence
\Cref{lem:consistency-gadget} gives
\(\LCS(\alpha_1|_F,\alpha_2|_F)=s/2\).

Applying \Cref{lem:canonical-score} with \(S_L=T\) and \(S_R=F\), we
obtain a permutation \(\pi^\star\) for which the sum of the four LCS
lengths is
\[
4M+s+m+\frac{s}{2}
=
4M+\frac{3s}{2}+m.
\]
Each input permutation has length \(M+s\). Therefore,
\[
\sum_{\rho\in\{L_1,L_2,R_1,R_2\}}
d_U(\pi^\star,\rho)
=
4(M+s)-\left(4M+\frac{3s}{2}+m\right)
=
\frac{5s}{2}-m.
\]

Conversely, suppose there exists a permutation \(\pi\) whose total
Ulam distance from \(L_1,L_2,R_1,R_2\) is at most \(5s/2-m\).
Equivalently, the sum of the corresponding four LCS lengths is at
least \(4M+3s/2+m\). By \Cref{lem:anchor-block}, we may replace
\(\pi\), without decreasing any of these LCS lengths, by a permutation
\(\tau_LZ\tau_R\) inducing a partition
\(\Sigma=S_L\mathbin{\dot\cup}S_R\).

By \Cref{lem:canonical-score}, the sum of the four LCS lengths is at
most
\[
4M+s
+\LCS(\beta_1|_{S_L},\beta_2|_{S_L})
+\LCS(\alpha_1|_{S_R},\alpha_2|_{S_R}).
\]
The final two terms are at most \(m\) and \(s/2\), respectively, by
\Cref{lem:clause-gadget,lem:consistency-gadget}. Since the sum is at
least \(4M+3s/2+m\), both upper bounds must be attained. Thus,
\(S_L\) contains an occurrence from every clause, and, for every
variable \(x\), the set \(S_R\) contains either all positive
occurrences of \(x\) or all negative occurrences of \(x\).

For each variable \(x\), choose one of the sets
\(\Sigma_x^+\) and \(\Sigma_x^-\) that is contained in \(S_R\), and
let \(F\) be the union of the chosen sets. Set
\(T:=\Sigma\setminus F\). By construction,
\(\Sigma=T\mathbin{\dot\cup}F\) is consistent. Moreover,
\(F\subseteq S_R\), and hence \(S_L\subseteq T\). Since \(S_L\)
intersects every clause, so does \(T\). Therefore,
\(\Sigma=T\mathbin{\dot\cup}F\) is a consistent and satisfying truth
partition. By \Cref{obs:truth-partition}, the formula \(\varphi\) is
satisfiable.
\end{proof}

Clearly, the reduction takes polynomial time. This completes the proof of \Cref{thm:main}.

\section{Proof of the Anchor-Block Lemma}
\label{sec:anchor-block-proof}

In this section, we give a proof of the anchor block lemma.

\begin{proof}[Proof of~\Cref{lem:anchor-block}]
Fix an arbitrary permutation \(\pi\) of
\(\Sigma\mathbin{\dot\cup}\Sigma_b\). We transform \(\pi\) into the
desired canonical form in two steps.

\paragraph{Step 1: sorting the anchor symbols.}
Let \(\widehat{\pi}\) be obtained from \(\pi\) by permuting only the
symbols of \(\Sigma_b\) so that they occur in the order
\(z_1,z_2,\ldots,z_M\). The symbols of \(\Sigma\) remain in their
original locations and hence retain their relative order. We claim
that, for each \(i\in\{1,2\}\),
\[
\LCS(\widehat{\pi},\gamma_iZ)
\geq
\LCS(\pi,\gamma_iZ)
\quad\text{and}\quad
\LCS(\widehat{\pi},Z\delta_i)
\geq
\LCS(\pi,Z\delta_i).
\]

For \(x\in\Sigma\), let \(a(x)\) denote the number of symbols from
\(\Sigma_b\) that precede \(x\) in \(\pi\). Since the reordering does
not change the locations occupied by symbols of \(\Sigma_b\), the
anchor symbols preceding \(x\) in \(\widehat{\pi}\) are precisely
\(z_1,\ldots,z_{a(x)}\), while those following \(x\) are precisely
\(z_{a(x)+1},\ldots,z_M\).

Fix \(i\in\{1,2\}\), and let \(\rho\) be a longest common subsequence
of \(\pi\) and \(\gamma_iZ\). If \(\rho\) contains no symbol of
\(\Sigma\), then \(|\rho|\leq M\), whereas \(Z\) is a common
subsequence of \(\widehat{\pi}\) and \(\gamma_iZ\) of length \(M\).

Otherwise, let \(x\) be the last symbol of \(\Sigma\) appearing in
\(\rho\). Since every symbol of \(\Sigma\) precedes every anchor symbol
in \(\gamma_iZ\), all anchor symbols appearing in \(\rho\) must follow
\(x\) in \(\pi\). Hence \(\rho\) contains at most \(M-a(x)\) anchor
symbols. Moreover,
\[
\bigl(\rho|_\Sigma\bigr)
z_{a(x)+1}z_{a(x)+2}\cdots z_M
\]
is a common subsequence of \(\widehat{\pi}\) and \(\gamma_iZ\). Its
length is at least \(|\rho|\), and therefore
\[
\LCS(\widehat{\pi},\gamma_iZ)
\geq
\LCS(\pi,\gamma_iZ).
\]
The proof for \(Z\delta_i\) is symmetric. Thus,
\[
\LCS(\widehat{\pi},Z\delta_i)
\geq
\LCS(\pi,Z\delta_i).
\]
It therefore suffices to consider the case in which the anchor symbols
occur in the order \(z_1,z_2,\ldots,z_M\). We may then write
\[
\widehat{\pi}
=
U_0z_1U_1z_2\cdots z_MU_M,
\]
where each \(U_j\) is a possibly empty string over \(\Sigma\).

\paragraph{Step 2a: analyzing \(\gamma_iZ\).}
Fix \(i\in\{1,2\}\). We claim that
\begin{equation}
\LCS(\widehat{\pi},\gamma_iZ)
=
M+\max_{0\leq t\leq M}
\bigl(
\LCS(U_0U_1\cdots U_t,\gamma_i)-t
\bigr).
\label{eq:cut-1}
\end{equation}

Indeed, for any \(t\in\{0,\ldots,M\}\), a common subsequence of
\(U_0U_1\cdots U_t\) and \(\gamma_i\) can be followed by
\(z_{t+1}z_{t+2}\cdots z_M\). Hence
\[
\LCS(\widehat{\pi},\gamma_iZ)
\geq
\LCS(U_0U_1\cdots U_t,\gamma_i)+M-t.
\]

For the reverse inequality, let \(\rho\) be any common subsequence of
\(\widehat{\pi}\) and \(\gamma_iZ\). If \(\rho\) contains no symbol
from \(\Sigma\), then \(|\rho|\leq M\), which is bounded by the
right-hand side of~\Cref{eq:cut-1} by taking \(t=0\).

Otherwise, suppose that the last symbol of \(\Sigma\) appearing in
\(\rho\) belongs to \(U_t\). Then \(\rho|_\Sigma\) is a common
subsequence of \(U_0U_1\cdots U_t\) and \(\gamma_i\). Moreover, every
symbol of \(\Sigma_b\) appearing in \(\rho\) must belong to
\(z_{t+1}z_{t+2}\cdots z_M\), so \(\rho\) contains at most \(M-t\)
such symbols. Therefore,
\[
|\rho|
\leq
\LCS(U_0U_1\cdots U_t,\gamma_i)+M-t.
\]
This proves~\Cref{eq:cut-1}.

Since \(\gamma_i\) has length \(s\), we have
\(\LCS(U_0U_1\cdots U_t,\gamma_i)\leq s\). Thus, when \(t>s\), the
quantity inside the maximum in~\Cref{eq:cut-1} is negative, whereas
for \(t=0\) it is nonnegative. Consequently, the maximum is attained
for some \(t\leq s\).

It follows that, for each \(i\in\{1,2\}\), there exists a longest
common subsequence of \(\widehat{\pi}\) and \(\gamma_iZ\) whose symbols
from \(\Sigma\) all belong to \(U_0U_1\cdots U_s\).

\paragraph{Step 2b: analyzing \(Z\delta_i\).}
A symmetric argument gives
\begin{equation}
\LCS(\widehat{\pi},Z\delta_i)
=
M+\max_{0\leq t\leq M}
\bigl(
\LCS(U_tU_{t+1}\cdots U_M,\delta_i)-(M-t)
\bigr).
\label{eq:cut-2}
\end{equation}

Since \(\delta_i\) has length \(s\), we have
\(\LCS(U_tU_{t+1}\cdots U_M,\delta_i)\leq s\). Thus, if
\(M-t>s\), then the quantity inside the maximum
in~\Cref{eq:cut-2} is negative, whereas for \(t=M\) it is
nonnegative. Consequently, the maximum is attained for some
\(t\geq M-s\). Since \(M=2s+1\), we have \(M-s=s+1\).

It follows that, for each \(i\in\{1,2\}\), there exists a longest
common subsequence of \(\widehat{\pi}\) and \(Z\delta_i\) whose symbols
from \(\Sigma\) all belong to
\(U_{s+1}U_{s+2}\cdots U_M\).

\paragraph{Step 2c: constructing the canonical permutation.}
Let \(S_L\) be the set of symbols appearing in
\(U_0U_1\cdots U_s\), and let \(S_R:=\Sigma\setminus S_L\). Define
\[
\tau_L:=U_0U_1\cdots U_s,
\qquad
\tau_R:=U_{s+1}U_{s+2}\cdots U_M,
\]
and set
\[
\pi^\star:=\tau_LZ\tau_R.
\]
We show that \(\pi^\star\) satisfies all four inequalities in the
statement of the lemma.

Fix \(i\in\{1,2\}\), and let \(\rho\) be a longest common subsequence
of \(\widehat{\pi}\) and \(\gamma_iZ\) whose symbols from \(\Sigma\)
all belong to \(U_0U_1\cdots U_s\), as guaranteed by Step~2a. Set
\(\rho_\Sigma:=\rho|_\Sigma\). The string \(\rho_\Sigma\) is a
subsequence of both \(\tau_L\) and \(\gamma_i\), since the construction
of \(\tau_L\) preserves the relative order of all symbols in
\(U_0U_1\cdots U_s\). It follows that
\(\rho_\Sigma Z\) is a common subsequence of
\(\pi^\star=\tau_LZ\tau_R\) and \(\gamma_iZ\). Moreover,
\[
|\rho_\Sigma Z|
=
|\rho_\Sigma|+M
\geq
|\rho|,
\]
because \(\rho\) contains at most \(M\) symbols from \(\Sigma_b\).
Therefore,
\[
\LCS(\pi^\star,\gamma_iZ)
\geq
\LCS(\widehat{\pi},\gamma_iZ).
\]

Similarly, let \(\rho\) be a longest common subsequence of
\(\widehat{\pi}\) and \(Z\delta_i\) whose symbols from \(\Sigma\) all
belong to \(U_{s+1}U_{s+2}\cdots U_M\), as guaranteed by Step~2b, and
again set \(\rho_\Sigma:=\rho|_\Sigma\). The string
\(\rho_\Sigma\) is a subsequence of both \(\tau_R\) and \(\delta_i\).
Consequently, \(Z\rho_\Sigma\) is a common subsequence of
\(\pi^\star\) and \(Z\delta_i\), and
\[
|Z\rho_\Sigma|
=
M+|\rho_\Sigma|
\geq
|\rho|.
\]
Hence
\[
\LCS(\pi^\star,Z\delta_i)
\geq
\LCS(\widehat{\pi},Z\delta_i).
\]

Finally, Step~1 established that sorting the symbols of \(\Sigma_b\)
does not decrease any of the relevant LCS lengths. Thus, for each
\(i\in\{1,2\}\),
\[
\LCS(\pi^\star,\gamma_iZ)
\geq
\LCS(\widehat{\pi},\gamma_iZ)
\geq
\LCS(\pi,\gamma_iZ),
\]
and
\[
\LCS(\pi^\star,Z\delta_i)
\geq
\LCS(\widehat{\pi},Z\delta_i)
\geq
\LCS(\pi,Z\delta_i).
\]
This proves the lemma.
\end{proof}

\bibliographystyle{alphaurl}
\bibliography{refs}

\end{document}